\documentclass[10pt]{article}
\usepackage[utf8]{inputenc}
\usepackage[margin=1in]{geometry}
\usepackage{rotating}
\usepackage{booktabs}
\usepackage{array}
\usepackage{amsmath, amsfonts}
\usepackage{amssymb}
\DeclareMathOperator*{\argmax}{arg\,max}
\usepackage{color}
\usepackage{hyperref}
\usepackage{bbm}
\usepackage{algorithm}
\usepackage[noend]{algpseudocode}
\usepackage{subcaption}
\usepackage{graphicx,nicefrac}
\usepackage{tikz}
\usepackage{adjustbox}
\usepackage{framed}
\usepackage[most]{tcolorbox}

\usepackage{mathtools}
\mathtoolsset{showonlyrefs}

\usepackage[
    backend=biber,
    style=ieee,
    citestyle=numeric-comp,
    sorting=none,
    giveninits=true,
    natbib,
    hyperref,
    maxbibnames=99,
    doi=true,isbn=false,url=true,eprint=false
]{biblatex}

\usepackage{amsthm}
\usepackage[normalem]{ulem}
\usepackage{soul}
\usepackage{tabularx}
\newtheorem{theorem}{Theorem} 
\newtheorem{lemma}{Lemma} 
\newtheorem{prop}{Proposition}
\newtheorem{cor}{Corollary}
\newtheorem{definition}{Definition}

\usepackage{chngcntr}
\usepackage{apptools}
\AtAppendix{\counterwithin{theorem}{section}}
\AtAppendix{\counterwithin{prop}{section}}

\usepackage{mathtools}

\newcommand{\D}{\mathcal{D}}

\newcommand{\hightrain}{\D_{h,\text{train}}}
\newcommand{\lowtrain}{\D_{l,\text{train}}}
\newcommand{\target}{\D_{\text{target}}}

\usepackage[]{color-edits}
\addauthor{mw}{magenta}

\makeatletter
\def\blfootnote{\xdef\@thefnmark{}\@footnotetext}
\makeatother

\title{Defection-Free Collaboration between Competitors \\ in a Learning System}

\begin{document}
\author{Mariel Werner$^1$ \and Sai Praneeth Karimireddy$^1$ \and Michael Jordan$^1$ }
\date{
$^1$Department of Electrical Engineering and Computer Sciences, \\ University of California, Berkeley}
\maketitle
\begin{abstract}
We study collaborative learning systems in which the participants are competitors who will defect from the system if they lose revenue by collaborating. As such, we frame the system as a duopoly of competitive firms who are each engaged in training machine-learning models and selling their predictions to a market of consumers. We first examine a fully collaborative scheme in which both firms share their models with each other and show that this leads to a market collapse with the revenues of both firms going to zero. We next show that one-sided collaboration in which only the firm with the lower-quality model shares improves the revenue of both firms. Finally, we propose a more equitable, \emph{defection-free} scheme in which both firms share with each other while losing no revenue, and we show that our algorithm converges to the Nash bargaining solution. 
\end{abstract}
\section{Introduction}
When the guarantees of a collaborative learning system are misaligned with the objectives of the learners, it can disincentivize participation and cause the participants to defect. Recent work \citep{donahue2021federated,blum2021OneAll,wu2022collaborate} examines the incentives that clients have to participate in or defect from a collaborative learning system. Such misalignment of incentives can arise in a number of ways. For example, \cite{karimireddy2022mechanisms} show that some clients might \emph{free-ride}, burdening other participants in the network with all the training work while contributing nothing. \citep{marfoq2021mixture,smith2021ditto,werner2023provably,ghosh2020efficient,mansour2020,sattler2019} show that if there is heterogeneity across clients' data distributions the global model returned by standard collaborative learning protocols might perform poorly for individual clients. To address the misalignment problem, \cite{han2023defections} propose an algorithm whose model updates guarantee that client losses degrade sufficiently from step to step to ensure that no client defects (albeit at some cost to the accuracy of the final global model).

In this paper, we take an economics-based view of the alignment problem, framing client \emph{utility/revenue} as the determining factor in defection. We frame clients as competitive firms who are selling their models' predictions to consumers and competing for market share. As in the standard collaborative learning protocol, the firms collaboratively train a global model, but if at any point in the process their revenue decreases, they defect from participation.

\paragraph{Motivating Example.} Consider two autonomous vehicle companies training self-driving models, each with initial access only to their own data. Further, suppose their individual data does not fully reflect the distribution on which the models must perform well at test time. For example, one company might have a lot of urban data and very little rural data and the other company the opposite. Clearly, if these companies combined their models, they could offer safer and better cars to consumers. However, by collaborating they might also lose their competitive advantage in the market, disincentivizing them from participating. Our objective is to design a collaboration scheme such that neither firm loses revenue, thus incentivizing participation.

\paragraph{Our Contributions.}  
We frame the collaborative learning system as a duopoly of competitive firms whose conditions for joining the system are to improve (or at least not lose) revenue, and we show that collaboration is possible under such conditions:
\begin{enumerate}
    \item We first show surprising outcomes of two possible collaboration schemes. When both firms contribute fully to the collaboration scheme, their model qualities improve maximally but their revenues go to zero. When only the low-quality firm contributes to the collaboration scheme, the model qualities and revenues of both firms improve.  
    \item We next design a defection-free algorithm which allows \emph{both} firms to contribute to the collaborative system without losing revenue at any step.
    \item We show that, except in trivial cases, our algorithm converges to the Nash bargaining solution. This is a significant result because we show that even when both firms focus myopically on improving their own revenues, a solution is reached that maximizes the joint surplus of the firms. 
\end{enumerate}

\subsection{Related Work}
Collaborative learning allows multiple clients to participate in the training of a global model without transmitting raw data \citep{mcmahan2017communication}. In this paper, we characterize the participants in a collaborative learning system as market competitors who will defect from collaboration if they lose revenue by participating. Competitive behavior of firms in markets is a well-established field of study in economics (see \cite{tirole1988industrial} for an overview). Particularly relevant to our work is competition in oligopolies \citep{cournot1838economics}. As in \cite{huang2023duopoly}, we structure our problem as a duopoly of competitive firms. Rather than focusing on incentivization of collaboration via revenue sharing between the firms, however, we study mechanisms that guarantee no revenue loss as we do in this paper. Another relevant line of work is that of \cite{tsoy2023datasharing}, who parameterize the data-sharing problem in terms of competition type (Bertrand \citep{bertrand1883competition} or Cournot \citep{cournot1838economics}), the number of data points each firm has, and the difficulty of the learning task, and give conditions on these parameters under which collaboration is profitable. As we do, they analyze various data-sharing schemes, such as full versus partial collaboration, and propose Nash bargaining \citep{nash1950bargaining} as a strategy for partial collaboration. Our work goes further in that we propose a federated optimization algorithm for reaching the Nash bargaining solution, guaranteeing no defections. 

\section{Collaborative Learning in an Oligopoly}

We frame the collaborative learning system as a duopoly (i.e., two firms) for simplicity, but all results can be extended to an oligopoly of more than two firms. 
 
Our setup is the following. Each firm possesses a model whose qualities are initially differentiated by classification accuracy on a target dataset. That is, one firm's model has low accuracy and the other firm's model has high accuracy on the target dataset. The consumers care about performance on the target distribution, which is different from the firms' individual data distributions. For example, in the autonomous vehicle example above, the target distribution would represent a variety of geographical locations, traffic instances, times of day/night, etc., while the individual distributions would not. Additionally we assume that the firms' individual distributions are complementary, so their combined data is distributed as the target distribution,  motivating the benefit of collaboration. Finally, we assume that, prior to collaboration, one firm has better initial model quality than the other (e.g., they have more training resources).  

A consumer has one of three options: 1) pay a higher price for the high-quality firm's model, 2) pay a lower price for the low-quality firm's model, or 3) buy neither model. We assume that all consumers would prefer the higher-quality model if the prices of both models were the same---that is, the firms' models are \emph{vertically differentiated}.
Consumers would be happiest if both firms collaborated fully since this would give them two maximally good models to choose from, but the initially high-quality firm would have sacrificed revenue in this scenario (we show this formally in Section \ref{sec:collaboration schemes}), causing it to defect. Based on this, our motivating question is: can we incentivize firms to join the collaboration scheme, thus benefiting consumers, while giving them no reason to defect due to revenue loss at any stage of the training process? We answer this question affirmatively. 

\subsection{Model}\label{subsec:duopoly model}
\subsubsection{Notation and assumptions}
\begin{enumerate} 
    \item A consumer's type corresponds to how much they value quality of prediction. We assume that consumer types are uniformly distributed on $\Theta = [0, 1]$, where consumer type $\theta=0$ places no value on quality and consumer type $\theta=1$ places maximal value on quality.
    \item We denote the low-quality firm's loss on its dataset with model parameters $x \in \mathcal{X}$ as $f(x;l) \in [0,1]$ and the high-quality firm's loss on its dataset as $f(x;h) \in [0,1]$. In the collaborative learning process, both firms want to solve the optimization problem
     \begin{align}\label{eq:shared objective}
        x^*=
        \arg\min_{x \in \mathcal{X}} 
        f(x), 
        &\phantom{{}=1} \text{where }
        f(x) \overset{\text{def}}{=} \frac{f(x;l)+f(x;h)}{2}.
     \end{align}
     That is, each firm wants to find the model which has minimal average loss across both firms' datasets. When the objective \eqref{eq:shared objective} is evaluated at the firms' models $x_l$ and $x_h$, we use the shorthand notation 
    \begin{align}
        f_l 
        &\overset{\text{def}}{=}
        \frac{f(x_l;l)+f(x_l;h)}{2},
        &\phantom{{}=1}
        f_h
        &\overset{\text{def}}{=}
        \frac{f(x_h;l)+f(x_h;h)}{2}.
    \end{align}
    Finally, we define model qualities $q(x)\overset{\text{def}}{=}1-f(x)$, $q_l\overset{\text{def}}{=}1-f_l$ and $q_h\overset{\text{def}}{=}1-f_h$.
    \item Consumers pay prices $p_{\nicefrac{l}{h}} \in [0,\infty)$ for the low/high-quality firm's model $x_{\nicefrac{l}{h}}$, where $p_l \leq p_h$.
\end{enumerate}
\subsubsection{Equilibrium quantities}
The following definition gives the consumer's utility.
\begin{definition}\label{def:consumer utility}[Consumer Utility] A type-$\theta$ consumer has utility
\begin{align}\label{eq:consumer utility}
U_c(\theta) = 
\begin{cases} 
\theta q_h - p_h & \textnormal{if it buys high-quality firm's model,} \\
\theta q_l - p_l & \textnormal{if it buys low-quality firm's model,} \\
0 & \textnormal{if it buys neither model}.
\end{cases}
\end{align}
\end{definition}
The consumer utilities in Definition \ref{def:consumer utility} induce the following demands for the firms.
\begin{lemma}[Consumer Demands]\label{lemma:consumer demands}
Given the utilities in Definition \ref{def:consumer utility}, 
\begin{enumerate}
    \item consumer demand for the low-quality firm is $D_l = \frac{p_h-p_l}{q_h-q_l} - \frac{p_l}{q_l}$, and
    \item consumer demand for the high-quality firm is $D_h = 1 - \frac{p_h-p_l}{q_h-q_l}$.
\end{enumerate}
\end{lemma}
\begin{proof}
    See Appendix \ref{proof:consumer demands lemma}.
\end{proof}

Using the consumer demands in Lemma \ref{lemma:consumer demands}, we can define the utilities of the firms. 
\begin{definition}\label{def:firm utility}[Firm Utility/Revenue]
The low/high firm's utility/revenue from selling its model is
\begin{align}\label{eq:firm utilities}
    U_{\nicefrac{l}{h}}(q_l, q_h, p_l, p_h) 
    &= 
    p_{\nicefrac{l}{h}}D_{\nicefrac{l}{h}}.
\end{align}
\end{definition}

At equilibrium, the firms will set prices $p_l$ and $p_h$ that maximize \eqref{eq:firm utilities}, yielding price-optimal utilities. 

\begin{lemma}[Equilibrium Prices and Utilities]\label{lemma:firm's price-optimal utility}
The optimal prices for the low and high firms are
\begin{align}\label{eq:optimal price low firm}
    p_l^* 
    = 
    \frac{q_l(q_h-q_l)}{4q_h-q_l},
    &\phantom{{}=1}
    p_h^* 
    = 
    \frac{2q_h(q_h-q_l)}{4q_h-q_l},
\end{align}
yielding price-optimal utilities
\begin{align}\label{eq:price-optimal utilities}
    U_l(q_l, q_h, p_l^*, p_h^*) 
    = \frac{q_lq_h(q_h-q_l)}{(4q_h-q_l)^2},
    \phantom{{}=1}
    U_h(q_l, q_h, p_l^*, p_h^*) = \frac{4q_h^2(q_h-q_l)}{(4q_h-q_l)^2}.
\end{align}
\end{lemma}
\begin{proof}
    See Appendix \ref{proof:firm's price-optimal utility lemma}.
\end{proof}
Going forward, we will use the shorthand $U_{\nicefrac{l}{h}}\overset{\text{def}}{=}U_{\nicefrac{l}{h}}(q_l,q_h,p_l^*,p_h^*)$.

The following proposition states how the firms' utilities vary with quality and is key in our analysis.
\begin{prop}[Relationship between utilities and qualities]\label{prop:relationship between U and q}
For $q_l \leq q_h$,
\begin{enumerate}
    \item $U_h$ is increasing in $q_h$,
    \item $U_h$ is decreasing in $q_l$,
    \item $U_l$ is increasing in $q_h$, and
    \item $U_l$ is increasing in $q_l$ for $q_l \leq \frac{4}{7}q_h$ and decreasing in $q_l$ otherwise. 
\end{enumerate}
\end{prop}
\begin{proof}
    See Appendix \ref{proof:relationship between U and q}
\end{proof}

In the next section, we examine various collaboration schemes between the firms and observe the impact on their revenues and model qualities.
\section{Collaboration Schemes}\label{sec:collaboration schemes}
To motivate our method, we describe two potential collaboration schemes between competitors that have sub-optimal and non-intuitive outcomes.
\paragraph{Sharing protocol.} As in standard federated learning protocols, we do not assume that the firms transmit their raw data to each other. Instead, firm A shares with firm B by evaluating the loss of firm B's model on firm A's data. Then firm A shares with firm B the loss, or the gradient of the loss, which allows firm B to optimize the objective \eqref{eq:shared objective}. These exchanges can happen either directly between the firms are through a trusted central coordinator.
\subsection{Notation and assumptions}
\begin{enumerate}
    \item $f(x;\nicefrac{l}{h})$ is convex and $L$-smooth in $x$. 
    \item We use $q_{\nicefrac{l}{h},t}$ and $f_{\nicefrac{l}{h},t}$ to refer to the firms' objectives when the model parameters are $x_{\nicefrac{l}{h},t}$, i.e., the model parameters at round $t$ of optimization.
    \item We define $\rho_t = \frac{q_{l,t}}{q_{h,t}}$, the ratio of the firms' model qualities at round $t$ of optimization.
    \item We assume model qualities can only improve or stay the same, not degrade.
\end{enumerate}

\subsection{Complete collaboration}


In this arrangement, both firms fully collaborate, sharing their models with each other and therefore obtaining identical-quality models. (Note that this algorithm is just FedAvg \citep{mcmahan2017communication}.) While this collaboration scheme is optimal for the consumer, giving them the choice of two maximally high-quality models, it drives both firms' utilities to zero. With identical-quality models, each firm will continually undercut the other's price by small amounts to capture the entire market share, eventually reaching equilibrium prices $p_l=p_h=0$.
\begin{lemma}[Firm revenues under Complete Collaboration]\label{lemma:complete sharing utilities}
    Under Complete Collaboration, the firms' equilibrium utilities are $U_l = U_h = 0$.
\end{lemma}

\begin{figure}[t!]
    \centering
    \begin{subfigure}{\textwidth}
        \centering
        \includegraphics[height=1.5in]{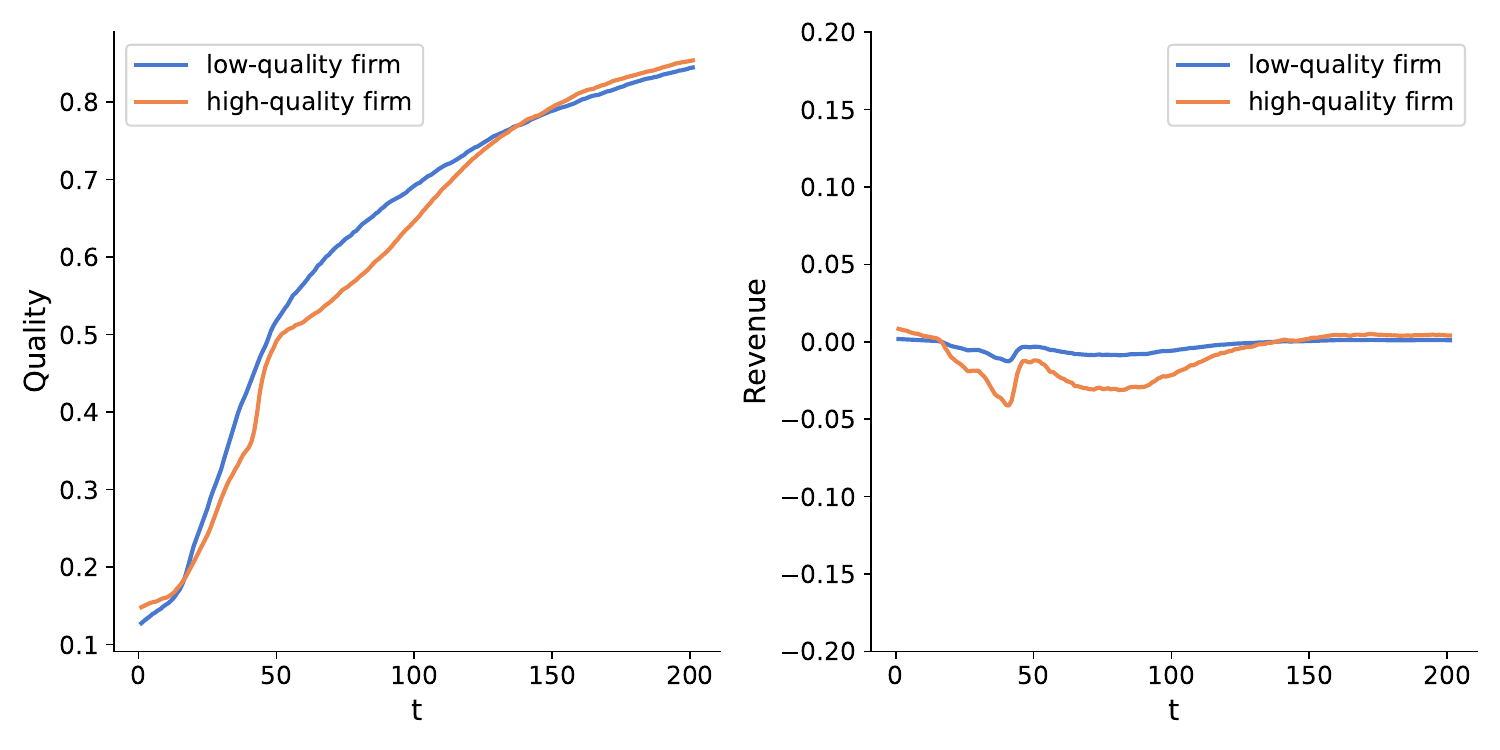}
        \label{subfig:mnist complete sharing}
    \end{subfigure}
    \caption{Performance of Complete Sharing scheme on MNIST. When both firms share with each other, their models converge to the same utility, driving their utilities to zero.}\label{fig:complete sharing utilities}
\end{figure}
Figure \ref{fig:complete sharing utilities} shows that when both firms' qualities increase freely in a Complete Collaboration scheme, their qualities both improve maximally, benefiting the consumer, but their utilities are driven to zero. Therefore, both firms have cause to defect from this collaboration scheme.
\subsection{One-sided collaboration}\label{subsec:one-sided collaboration}


In One-sided Collaboration, one firm shares its model while the other doesn't. There are two possibilities.
\paragraph{Only high-quality firm shares.} From Proposition \ref{prop:relationship between U and q}, the high-quality firm's revenue increases in $q_h$ but decreases in $q_l$. Therefore, if the quality of $x_h$ does not increase sufficiently to compensate for the increase in quality of $x_l$, the high-quality firm will lose revenue, causing it to defect. (In the proof of Proposition \ref{prop:non-decreasing utilities}, we give this increase-threshold precisely.) In our problem setup, the individual firms' data distributions are different than target distribution on which the qualities of their models are evaluated. Therefore, if the low-quality firm benefits from the high-quality firm's model, its performance on the target distribution will outpace the high-quality firm, which is limited to training on its own data. Figure \ref{subfig:mnist one-sided sharing high shares} gives an example of this outcome. Due to collaboration, the low-quality firm's model out-performs the high-quality firm's model, causing the high-quality firm's revenue to decrease. 
\paragraph{Only low-quality firm shares.}
From Proposition \ref{prop:relationship between U and q}, both firms' utilities increase in $q_h$. Therefore, both firms will increase their revenue if the low-quality firm shares its model with the high-quality firm. Figure \ref{subfig:mnist one-sided sharing low shares} depicts the outcome of this collaboration scheme. Over time, both firms' revenues increase. While this arrangement is defection-free, the low-quality firm is stuck with its own data, causing it to potentially have lower revenue that it would under a more equitable scheme. To address this, we next propose a defection-free scheme in which \emph{both} firms participate in collaboration without losing revenue. 
\begin{figure}[t!]
    \centering
    \begin{subfigure}{0.45\textwidth}
        \centering
        \includegraphics[height=1.5in]{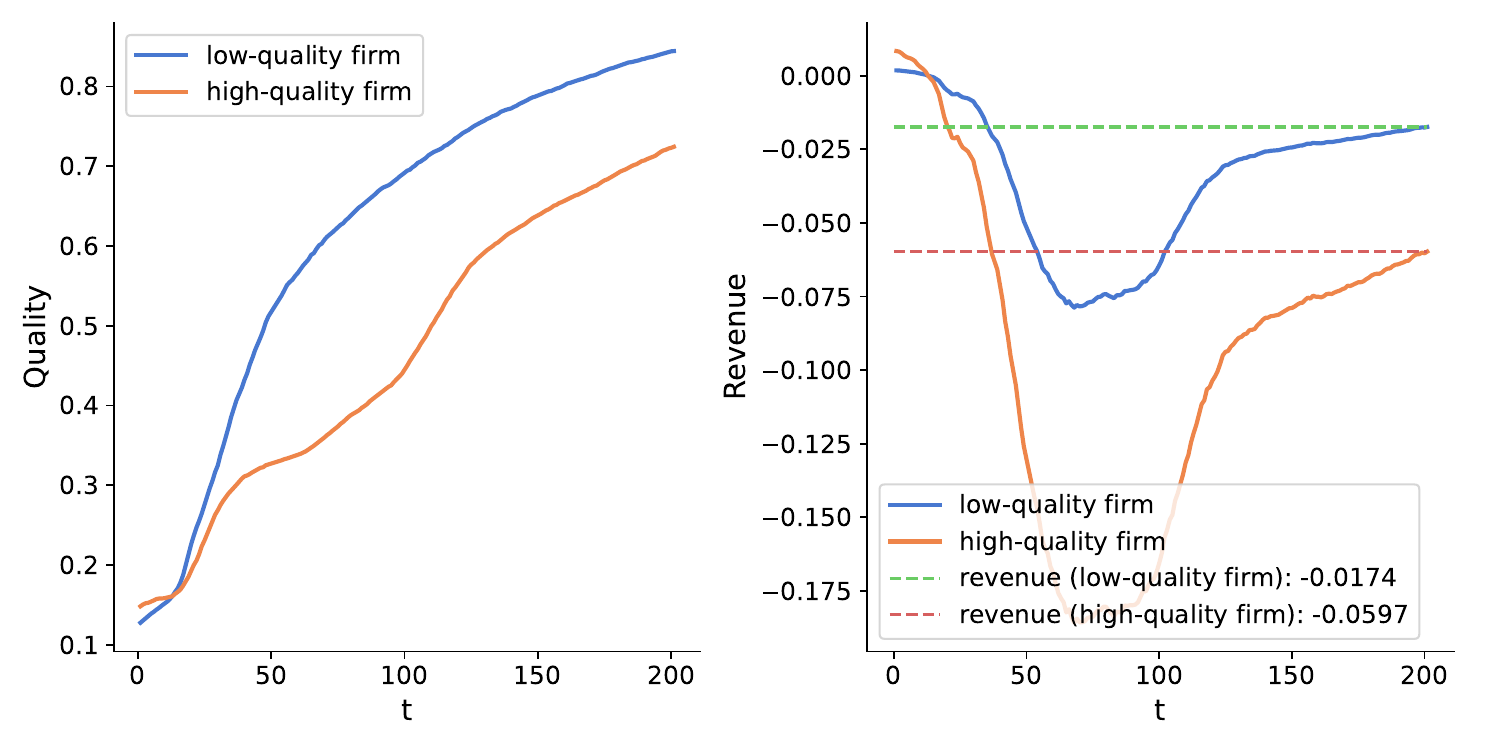}\caption{Only high-quality firm shares.}
        \label{subfig:mnist one-sided sharing high shares}
    \end{subfigure}
    ~
    \begin{subfigure}{0.45\textwidth}
        \centering
        \includegraphics[height=1.5in]{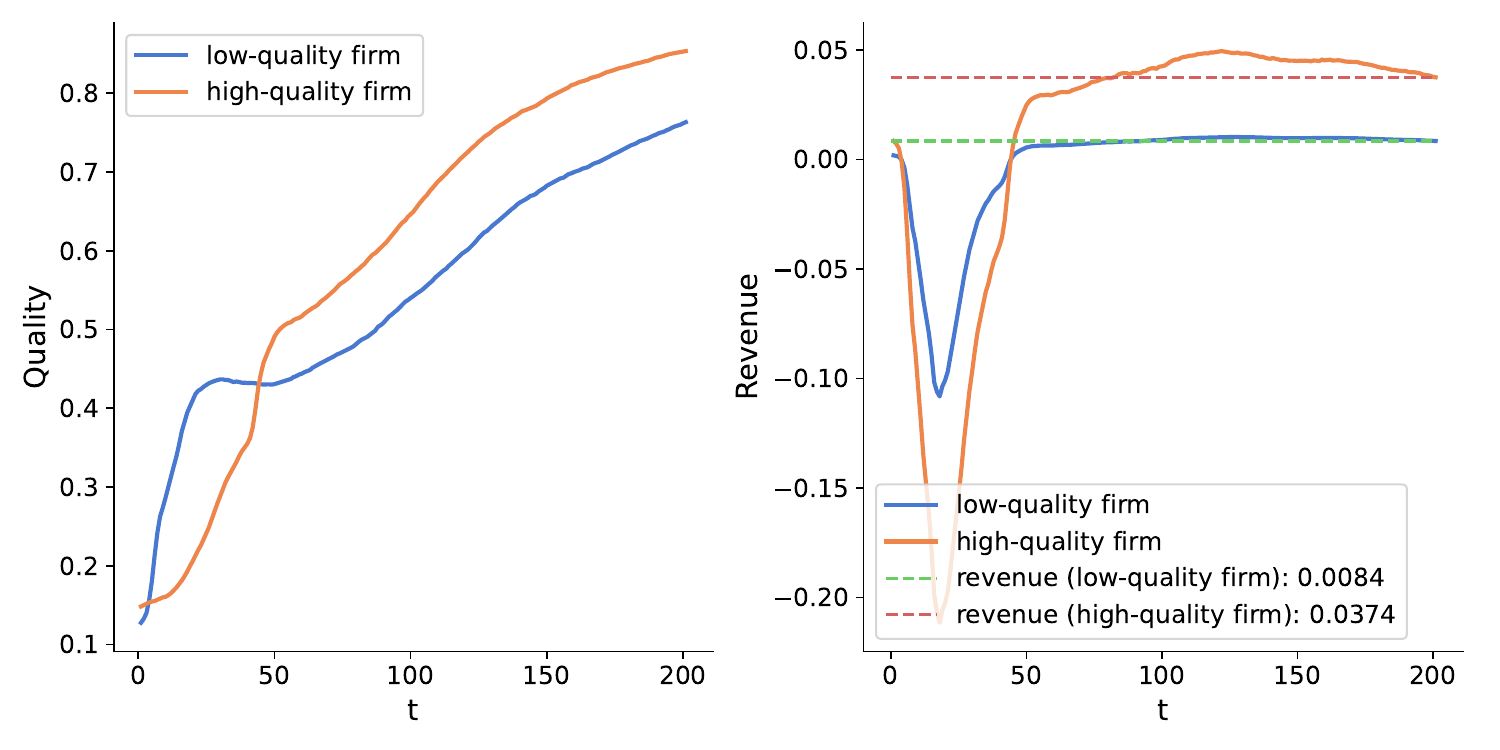}\caption{Only low-quality firm shares.}
        \label{subfig:mnist one-sided sharing low shares}
    \end{subfigure}
    \caption{Performance of One-sided Sharing schemes on MNIST. When only the high-quality firm shares, the high-quality firm's revenue becomes negative. When only the low-quality firm shares, both firms have positive, but less, revenue than with our collaboration scheme (Figure \ref{fig:nash sharing utilities}).}\label{fig:one-sided sharing utilities}
\end{figure}

\section{Defection-Free Collaborative Learning}
In this section, we introduce our method, Defection-Free Collaborative Learning (Defection-Free CL). Our objectives in designing this algorithm are that
\begin{enumerate}
    \item For all starting values $(q_{l,0}, q_{h,0})$, neither firm's revenue decreases at any round, and
    \item The algorithm converges to the Nash bargaining solution, which we denote $(q_l^*, q_h^*)$. (See Section \ref{sec:nash bargaining solution}.)
\end{enumerate}
The first objective ensures that the algorithm is defection-free. The second seeks a point of convergence that maximizes the joint surplus of the firms. In Section \ref{subsec:theory and analysis}, we show that Algorithm \ref{alg:defection-free cl} achieves 1) entirely and achieves 2) for a large range of starting conditions. Before describing our algorithm, we first motivate the Nash bargaining solution as a suitable convergence goal for our problem setting. 

\subsection{Nash bargaining}\label{sec:nash bargaining solution}
In cooperative bargaining, agents determine how to share a surplus amongst themselves. If negotiations fail, each agent is guaranteed some fixed surplus, known as the \emph{disagreement point}. A typical application of bargaining involves deciding how to split a firm's profits amongst its employees. The bargaining framework is suitable for our purposes because the firms must agree how to share a ``surplus of quality'' (i.e., set model qualities relative to each other) so that neither firm's revenue decreases at any one round. 

An important framework in cooperative bargaining is Nash bargaining \citep{nash1950bargaining}, a two-person bargaining scheme, which solves for 
\begin{align}\label{eq:Nash bargaining problem}
    (q_l^*, q_h^*) 
    = 
    \argmax_{(q_l, q_h)} \quad & N(q_l, q_h, q_{l,0}, q_{h,0})\\
    \textrm{s.t.} \quad & U_l(q_l, q_h) \geq U_l(q_{l,0}, q_{h,0})\\
    \quad & U_h(q_l, q_h) \geq U_h(q_{l,0}, q_{h,0}),
\end{align}
where
\begin{align}\label{eq:Nash bargaining objective}
    N(q_l, q_h, q_{l,0}, q_{h,0}) 
    \overset{\text{def}}{=} 
    (U_l(q_l, q_h) - U_l(q_{l,0}, q_{h,0}))(U_h(q_l, q_h) - U_h(q_{l,0}, q_{h,0})),
\end{align}
and $(q_{l,0}, q_{h,0})$ are the initial model qualities of the firms. The \emph{Nash bargaining solution}, $(q_l^*, q_h^*)$, maximizes the product of the \emph{improvement} in the firms' utilities. Therefore, unlike one-sided collaboration, the Nash objective rewards improvement in the low-quality firm's utility as well as the high-quality firm's utility. In Nash bargaining, the \emph{disagreement point} $(q_{l,0},q_{h,0})$ determines the surplus for the parties if negotiations fall apart. In our setting, if either firm defects from collaboration, both firms retain their current model qualities. Going forward, we use $N(q_l,q_h)$ as shorthand for $N(q_l,q_h,q_{l,0},q_{h,0})$. The Nash bargaining solution $(q_l^*, q_h^*)$ has four important properties: 1) it is invariant to affine transformation of the utility functions, 2) it is pareto efficient, 3) it is symmetric, and 4) it is independent of irrelevant alternatives. In fact, the point $(q_l, q_h)$ with these four properties is uniquely the Nash bargaining solution.  

The next proposition shows that $q_h^*$ is equivalent to the high-quality firm's maximal quality.
\begin{prop}[Equivalence between maximal quality and the Nash bargaining solution]
\begin{align}
    q_h^* = \max_{x \in \mathcal{X}} q(x).
\end{align}
\end{prop}
\begin{proof}
From Proposition \ref{prop:relationship between U and q}, $\frac{\partial U_h}{\partial q_h}$ and $\frac{\partial U_l}{\partial q_h}$ are both non-negative for all $q_l \leq q_h$, and consequently $\frac{\partial N(q_l,q_h)}{\partial q_h} \geq 0$ for all $q_l \leq q_h$. This means that for any $q_l$, the $N(q_l, q_h)$ can always be improved by increasing $q_h$. Therefore, $q_h^*$ is necessarily $\max_{x \in \mathcal{X}} q(x)$. 
\end{proof}

Section \ref{subsec:one-sided collaboration} shows there is a defection-free scheme in which the low-quality firm shares but the high-quality firm doesn't. In Algorithm \ref{alg:defection-free cl}, we give a way for both firms to contribute to collaboration with neither firm losing revenue at any step. Due to the more equitable design of this collaboration scheme, its dynamics mirror those of Nash bargaining which maximizes the joint surplus of the participants. 

The difficulty of designing Algorithm \ref{alg:defection-free cl} is that, in order to reach $(q_l^*, q_h^*)$ without decreasing revenues at any step, neither firm can improve its quality too much in a given step. Given an increase in the high-quality firm's quality $q_{h,t-1} \rightarrow q_{h,t}$, the low-quality firm can only improve by some limited amount without decreasing the high-firm's revenue (since $U_h$ is decreasing in $q_l$ by Prop. \ref{prop:relationship between U and q}). Because of this capped permissible improvement for the low-quality firm, if the high-quality firm converges to $q_h^*$ too quickly, the low-quality firm will never reach $q_l^*$. 

We describe the key steps of Algorithm \ref{alg:defection-free cl}. We also assume that, prior to the algorithm, both firms have saturated training on their own datasets and will only update their models collaboratively going forward. Since $U_l$ and $U_h$ both increase in $q_h$, the low-quality firm should always share with the high-quality firm. Step 4 ensures this, where the high-quality firm has access to the low-quality firm's loss on its model $x_{h,t-1}$ when updating. As we show in Section \ref{subsec:theory and analysis}, in order to converge to the Nash bargaining solution, the low-quality firm should not update if $q_{l,t} \geq q_l^*$ or $\rho_{t-1} > \rho^*$. Step 7 ensures this. Since $U_h$ decreases in $q_l$, the low-quality firm cannot improve its model beyond a certain threshold before the high-quality firm loses revenue. This threshold $\hat{q}_{l,t}$ is computed in Step 8, and in Steps 9-11, the high-quality firm will only collaborate if the collaborative updates to the low-quality firm's model do not improve its quality beyond $\hat{q}_{l,t}$. 

\begin{algorithm}[!t]
\caption{Defection-Free Collaborative Learning}\label{alg:defection-free cl}
\textbf{Input:} Low-quality model: $x_{l,0}$. High-quality model: $x_{h,0}$.
\newline
\textbf{Note:} We assume both firms are trusted parties and will honestly exchange information. For example, to perform the necessary computations, the high-quality firm requires $x_l$ and $\nabla f(x_h;l)$ from the low-quality firm, and the low-quality firm requires $x_h$, $\nabla f(x_l;h)$, $f(x_h;h)$, and $f(x_l;h)$ from the high-quality firm. 
\begin{algorithmic}[1]
\For {$t \in [T]$}
\State \underline{\textbf{High-quality Model Update}}
\State Set $\alpha_{h,t} \leq \frac{1}{L}$.
\State Update: $x_{h,t} = x_{h,t-1} - \alpha_{h,t}\nabla_{x_{h,t-1}} f_{h,t-1}$.
\State \underline{\textbf{Low-quality Model Update}}
\State $x_{l,t} = x_{l,t-1}$.
\If {$q_{l,t} < q_l^*$ and $\frac{q_{l,t}}{q_{h,t}}\leq \rho^*=\frac{q_l^*}{q_h^*}$}
\State Compute: $\hat{q}_{l,t} = B\bigg(\rho_{t-1},\frac{q_{h,t}}{q_{h,t-1}}\bigg)q_{h,t}$,
where
\begin{align}
    B(a,b)
    &\overset{\text{def}}{=}
    4-\frac{(4-a)^2}{2(1-a)}\bigg(b-\sqrt{b^2 - \frac{12(1-a)}{(4-a)^2}b}\bigg).
\end{align}
\While {$q_{l,t} \leq \hat{q}_{l,t}$}
\State Set: $\alpha_{l,t}$.
\State Update: $x_{l,t} \leftarrow x_{l,t} - \alpha_{l,t}\nabla_{x_{l,t}} f_{l,t}$
\EndWhile
\EndIf
\EndFor
\State{\bfseries Output: $x_{l,T}, x_{h,T}$}
\end{algorithmic}
\end{algorithm}


In the next section we prove the two key properties of Defection-Free Collaborative Learning: 1) it guarantees the firms non-decreasing revenue at every step, and 2) it converges to the Nash bargaining solution for all but trivial starting conditions.

\subsection{Theory and analysis}\label{subsec:theory and analysis}
  
The following proposition shows that Algorithm \ref{alg:defection-free cl} is defection-free.
\begin{prop}[Non-decreasing revenues]\label{prop:non-decreasing utilities}
There exist learning rate schedules $\{\alpha_{l,t}\}_t$ and $\{\alpha_{h,t}\}_t$ such that at no step of Algorithm \ref{alg:defection-free cl} does either firm's revenue decrease.
\end{prop}
\begin{proof}
    See Appendix \ref{proof:non-decreasing utilities prop}.
\end{proof}

We next examine starting conditions for which Algorithm \ref{alg:defection-free cl} converges to the Nash bargaining solution. Proposition \ref{prop:impossibility of convergence} gives a trivial starting condition for which it does not converge. 

\begin{prop}[Impossibility of convergence to the Nash bargaining solution]\label{prop:impossibility of convergence}
If $q_{l,0} > q_l^*$, then Algorithm \ref{alg:defection-free cl} cannot converge to $(q_l^*,q_h^*)$.
\end{prop}
\begin{proof}
Since firms do not degrade their model quality, the low-quality firm cannot converge to $q_l^*$. 





\end{proof}

In the next proposition, we show that for all other starting conditions, Algorithm \ref{alg:defection-free cl} converges to $(q_l^*,q_h^*)$. Our key insight in the proof of this proposition is that if the high-quality firm converges too quickly to $q_h^*$, the low-quality firm will not be able to make sufficient progress towards $q_l^*$ without violating the no-revenue-loss condition. Therefore, we must design a learning rate schedule for the high-quality firm $\{\alpha_{h,t}\}_t$ such that convergence to $q_h^*$ is properly paced.


 
\begin{prop}[Convergence to the Nash bargaining solution]\label{prop:sufficient conditions for convergence}
If $q_{l,0} \leq q_l^*$, then there exist learning rate schedules $\{\alpha_{l,t}\}_{t=1}^T$ and $\{\alpha_{h,t}\}_{t=1}^T$ such that after $T$ rounds Algorithm \ref{alg:defection-free cl} converges to $(q_l^*,q_h^*)$.
\end{prop}
\begin{proof}
    See Appendix \ref{proof:sufficient conditions for convergence prop}.
\end{proof}

Proposition \ref{prop:sufficient conditions for convergence} shows that even when both firms myopically attend to improving their own revenues, Algorithm \ref{alg:defection-free cl} converges to the Nash bargaining solution which maximizes joint surplus. The following theorem gives the rate of convergence to the Nash bargaining solution for convex and $L$-smooth losses. 
\begin{theorem}[Convergence Rate of Defection-Free Collaborative Learning]\label{thm:convergence rate}
Suppose $q_{l,0} \leq q_l^*$. Then running Algorithm \ref{alg:defection-free cl} for $T$ rounds ensures
\begin{align}\label{eq:nash convergence}
     N(q_l^*, q_h^*) - N(q_{l,T}, q_{h,T})
     &\lesssim
     \frac{\|x_{h,0}-x_h^*\|^2}{\sum_{t=1}^T \alpha_{h,t}} + |\rho^* - \rho_T|.
\end{align}
\end{theorem}
\begin{proof}
    See Appendix \ref{proof:convergence rate theorem}.
\end{proof}
The first term in the bound \eqref{eq:nash convergence} shows that the convergence rate to the Nash bargaining solution is determined by the rate at which $q_h$ converges to $q_h^*$.

The following corollary shows the rate at which the $|\rho^* - \rho_T|$ term in Theorem \ref{thm:convergence rate} vanishes with $T$.
\begin{cor}\label{cor:simplified convergence bound}
Suppose $q_{l,0} \leq q_l^*$. Running Algorithm \ref{alg:defection-free cl} for $T \gtrsim \frac{L\|x_{h,0}-x_h^*\|^2}{\epsilon}$ rounds ensures that
\begin{align}\label{eq:simplified convergence bound}
     N(q_l^*, q_h^*) - N(q_{l,T}, q_{h,T})
     &\lesssim
     \frac{\|x_{h,0}-x_h^*\|^2}{\sum_{t=1}^T \alpha_{h,t}} + (4-5\rho^*)\log\bigg(\frac{q_h^*}{q_h^*-\epsilon}\bigg).
\end{align}
\end{cor}
\begin{proof}
    See Appendix \ref{proof:simplified convergence bound corollary}.
\end{proof}
\section{Experiments}
All algorithms in our \href{https://github.com/mwerner28/datasharing}{experiments} are implemented with PyTorch \citep{paszke2019pytorch}. Our general experimental setup is the following. We construct three datasets: the low-quality firm's dataset $\lowtrain$, the high-quality firm's dataset $\hightrain$, and a common test set for both firms $\target$. The datasets are constructed such that $\lowtrain \not\sim \target$ and $\hightrain \not\sim \target$, but $\lowtrain \cup \hightrain \sim \target$, i.e., neither firm's individual distribution matches the target distribution, but their combined datasets are distributed as the target distribution, incentivizing them to share. We use cross-entropy loss and PyTorch's built-in SGD optimizer for all experiments.


\paragraph{MNIST} We use a LeNet-5 model \citep{lecun1998learning}, set $|\lowtrain|=|\hightrain|=1000$, and use the MNIST test set as $\target$. We construct $\lowtrain$ so that $\hat{F}(5)=0.8$ and $\hightrain$ so that $\hat{F}(5)=0.2$, where $\hat{F}$ is the empirical CDF over the label space. We train the high-quality firm's model for 10 initial epochs, and for all models and experiments set the learning rate to $0.01$.
\paragraph{Defection-Free Collaborative Learning (Figure \ref{fig:nash sharing utilities}).} Since the low-quality firm shares with the high-quality firm, the high-quality firm improves maximally. The high-quality firm only shares with the low-quality firm to the extent that neither firm's revenue decreases. Under this sharing scheme, we see in the first figure that both firms' qualities increase, and the ratio of their qualities converges to the optimal ratio. The second figure shows that revenues increase (do not decrease), and notably their revenues reach a higher level than under One-sided Collaboration (Section \ref{subsec:one-sided collaboration}). Finally, the last figure shows that the Nash bargaining objective approaches its maximal value, showing convergence to the Nash bargaining solution.
\begin{figure}[t!]
    \centering
    \begin{subfigure}{\textwidth}
        \centering
        \includegraphics[height=1.5in]{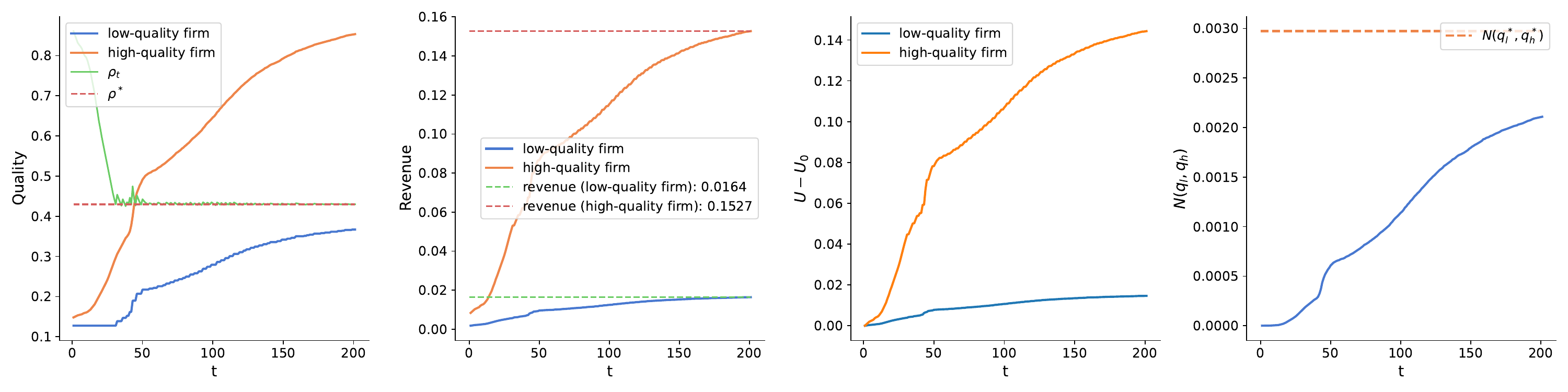}
        \label{subfig:mnist nash sharing}
    \end{subfigure}
    \caption{Performance of Defection-Free FL on MNIST. Both firms' qualities increase (figure 1), their revenues increase and approach a higher level than under One-sided Collaboration (figure 2), and the firms' qualities approach the Nash bargaining solution (figure 4).}\label{fig:nash sharing utilities}
\end{figure}
\newpage
\printbibliography
\appendix
\section{Proofs}
\subsection{Proofs for Section \ref{subsec:duopoly model}}
\begin{proof}[Proof of Lemma \ref{lemma:consumer demands}]\label{proof:consumer demands lemma}
Let $\hat{\theta}_l$ be the type of the consumer who is indifferent between buying from the low-quality firm and not buying at all. Then, based on the consumer's utility function \eqref{eq:consumer utility},
\begin{align}\label{eq:low-cutoff}
    \hat{\theta}_lq_l-p_l = 0.
\end{align}
Let $\hat{\theta}_h$ be the type of the consumer who is indifferent between buying from the high-quality firm and low-quality firm. Then, from \eqref{eq:consumer utility},
\begin{align}\label{eq:high-cutoff}
    \hat{\theta}_hq_l-p_l = \hat{\theta}_hq_h-p_h.
\end{align}
Therefore any consumer with type $\theta \in [\hat{\theta}_l, \hat{\theta}_h)$ will buy from the low-quality firm and any consumer with type $\theta \in [\hat{\theta_h}, 1]$ will buy from the high-quality firm, giving demands $D_l=\hat{\theta}_h-\hat{\theta}_l$ and $D_h=1-\hat{\theta_h}$. Solving \eqref{eq:low-cutoff} and \eqref{eq:high-cutoff} for $\hat{\theta}_l$ and $\hat{\theta}_h$ completes the proof.
\end{proof}

\begin{proof}[Proof of Lemma \ref{lemma:firm's price-optimal utility}]\label{proof:firm's price-optimal utility lemma}
From Lemma \ref{lemma:consumer demands}, the demand for the low-quality firm is $D_l = \frac{p_h-p_l}{q_h-q_l}-\frac{p_l}{q_l}$, yielding low-quality firm utility
\begin{align}\label{eq:low-utility}
    U_l = p_l\bigg(\frac{p_h-p_l}{q_h-q_l}-\frac{p_l}{q_l}\bigg). 
\end{align}
To maximize its utility, the low-quality firm sets price
\begin{align}
    p_l^* 
    &= \argmax_{p_l}\frac{\partial U_l}{\partial p_l}\\
    &= \argmax_{p_l} \bigg(\frac{p_h-2p_l}{q_h-q_l} - \frac{2p_l}{q_l}\bigg) \\
    &=
    \frac{q_lp_h}{2q_h}\label{eq:opt-low-price}.
\end{align}
Similarly, demand for the high-quality firm is $D_h = 1-\frac{p_h-p_l}{q_h-q_l}$, yielding high-quality firm utility
\begin{align}\label{eq:high-utility}
    U_h = p_h\bigg(1-\frac{p_h-p_l}{q_h-q_l}\bigg).
\end{align}
To maximize its utility, the high-quality firm sets price
\begin{align}
    p_h^*
    &=
    \argmax_{p_h}\frac{\partial U_h}{\partial p_h}\\
    &=
    \argmax_{p_h}\bigg(1-\frac{2p_h-p_l}{q_h-q_l}\bigg)\\
    &=
    \frac{p_l+(q_h-q_l)}{2}\label{eq:opt-high-price}.
\end{align}
Resolving \eqref{eq:opt-low-price} and \eqref{eq:opt-high-price} yields
\begin{align}\label{eq:resolve-opt-low-price}
    p_l^* = \frac{q_l(q_h-q_l)}{4q_h-q_l}
\end{align}
and
\begin{align}\label{eq:resolve-opt-high-price}
    p_h^* = \frac{2q_h(q_h-q_l)}{4q_h-q_l}.
\end{align}
Finally, evaluating \eqref{eq:low-utility} and \eqref{eq:high-utility} at the optimal prices \eqref{eq:resolve-opt-low-price} and \eqref{eq:resolve-opt-high-price} yields the price-optimal utilities \eqref{eq:price-optimal utilities}.
\end{proof}

\begin{proof}[Proof of Proposition \ref{prop:relationship between U and q}]\label{proof:relationship between U and q}
The proposition follows from observing the partial derivatives of the firms' utility functions. For $q_l \leq q_h$,
\begin{align}
    \frac{\partial U_h}{\partial q_h} = \frac{4q_h(4q_h^2-3q_hq_l+2q_l^2)}{(4q_h-q_l)^3}
    \geq 
    0,
\end{align}
\begin{align}
    \frac{\partial U_l}{\partial q_h} = \frac{q_l^2(2q_h+q_l)}{(4q_h-q_l)^3}
    \geq
    0,
\end{align}
\begin{align}
    \frac{\partial U_l}{\partial q_l}
    &=
    \frac{q_h^2(4q_h-7q_l)}{(4q_h-q_l)^3}
    \begin{cases}
        \geq 0 & \text{if } q_l \leq \frac{4}{7}q_h\\
        < 0 & \text{if } q_l > \frac{4}{7}q_h
    \end{cases}
\end{align}
and
\begin{align}
    \frac{\partial U_h}{\partial q_l}
    &=
    -\frac{4q_h^2(2q_h+q_l)}{(4q_h-q_l)^3}
    \leq
    0.
\end{align}
Figure \ref{fig:utilities} provides a graphical representation of this proposition.
\begin{figure}[t!]
    \centering
    \includegraphics[height=3in]{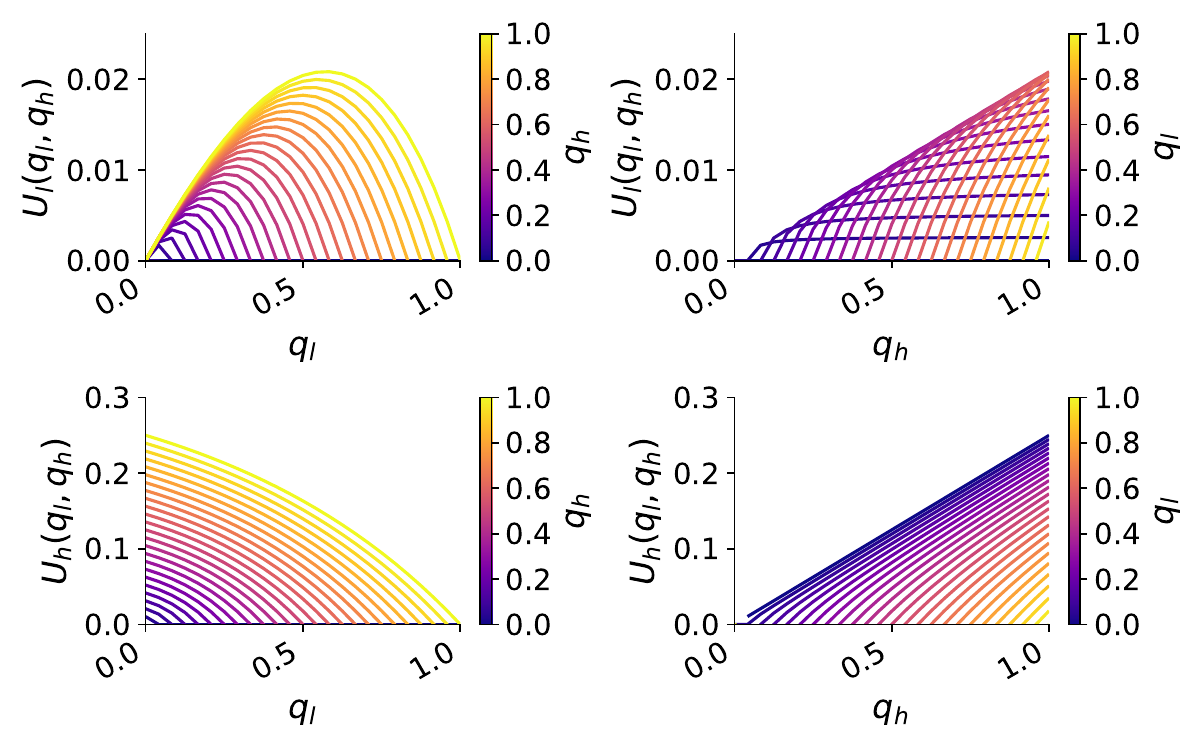}
    \caption{This figure shows how the firms' utilities vary with model quality. $U_l$ and $U_h$ are both increasing in $q_h$, $U_h$ is decreasing in $q_l$, and $U_l$ is increasing in $q_l$ for $q_l \leq \frac{4q_h}{7}$ and decreasing in $q_l$ otherwise.}\label{fig:utilities}
\end{figure}
\end{proof}

\subsection{Proofs for Section \ref{subsec:theory and analysis}}
\begin{proof}[Proof of Proposition \ref{prop:non-decreasing utilities}]\label{proof:non-decreasing utilities prop}
Suppose that at round $t$, given current qualities $q_{l,t-1}$ and $q_{h,t-1}$, the high-quality firm improves to $q_{h,t}$. Then, in order for neither firm to lose revenue, $q_{l,t}$ must be such that
\begin{align}
    \frac{4q_{h,t}^2(q_{h,t}-q_{l,t})}{(4q_{h,t}-q_{l,t})^2} 
    &\geq 
    \frac{4q_{h,t-1}^2(q_{h,t-1}-q_{l,t-1})}{(4q_{h,t-1}-q_{l,t-1})^2}\label{eq:high utility increase}
\end{align}
and
\begin{align}
    \frac{q_{l,t}q_{h,t}(q_{h,t}-q_{l,t})}{(4q_{h,t}-q_{l,t})^2}
    &\geq
    \frac{q_{l,t-1}q_{h,t-1}(q_{h,t-1}-q_{l,t-1})}{(4q_{h,t-1}-q_{l,t-1})^2}\label{eq:low utility increase}.
\end{align}
Rearranging terms, \eqref{eq:high utility increase} can be written as an inequality involving a convex quadratic of $q_{l,t}$:
\begin{align}
    &[4q_{h,t-1}^2(q_{h,t-1}-q_{l,t-1})]q_{l,t}^2\\
    &\phantom{{}=1} + 
    [4(4q_{h,t-1}-q_{l,t-1})^2q_{h,t}^2-32q_{h,t-1}^2(q_{h,t-1}-q_{l,t-1})q_{h,t}]q_{l,t}\\
    &\phantom{{}=1} +
    [64q_{h,t-1}^2(q_{h,t-1}-q_{l,t-1})q_{h,t}^2 - 4(4q_{h,t-1}-q_{l,t-1})^2q_{h,t}^3]< 0.
\end{align}
The right-most root of this quadratic is
\begin{align}
    q^h_{l,t} = 
    4q_{h,t} - \frac{(4-\rho_{t-1})^2}{2(1-\rho_{t-1})}
    \bigg(\frac{q_{h,t}^2}{q_{h,t-1}}-\sqrt{\frac{q_{h,t}^4}{q_{h,t-1}^2}-\frac{12(1-\rho_{t-1})}{(4-\rho_{t-1})^2}\frac{q_{h,t}^3}{q_{h,t-1}}}\bigg).
\end{align}
Similarly, \eqref{eq:low utility increase} can be written as an inequality involving a convex quadratic of $q_{l,t}$:
\begin{align}
    &[q_{l,t-1}q_{h,t-1}(q_{h,t-1}-q_{l,t-1})+(4q_{h,t-1}-q_{l,t-1})^2q_{h,t}]q_{l,t}^2\\
    &\phantom{{}=1} + 
    [-8q_{l,t-1}q_{h,t-1}(q_{h,t-1}-q_{l,t-1})q_{h,t}-(4q_{h,t-1}-q_{l,t-1})^2q_{h,t}^2]q_{l,t}\\
    &\phantom{{}=1} +
    [16q_{l,t-1}q_{h,t-1}(q_{h,t-1}-q_{l,t-1})q_{h,t}^2]< 0.
\end{align}
The right-most root of this quadratic is
\begin{align}
    q^l_{l,t} = 
    \frac{8(1-\rho_{t-1})\rho_{t-1}q_{h,t-1}+(4-\rho_{t-1})^2q_{h,t}+(4-\rho_{t-1})\sqrt{(4-\rho_{t-1})^2q_{h,t}^2-48\rho_{t-1}(1-\rho_{t-1})q_{h,t-1}q_{h,t}}}{2((1-\rho_{t-1})\rho_{t-1}q_{h,t-1}+(4-\rho_{t-1})^2q_{h,t})}.
\end{align}
It can be verified with graphing software that for all feasible parameters, $q_{l,t}^h \leq q_{l,t}^l$. Therefore, the low-quality firm can improve its quality to at most
\begin{align}
    \hat{q}_{l,t}
    &=
    4q_{h,t} - \frac{(4-\rho_{t-1})^2}{2(1-\rho_{t-1})}
    \bigg(\frac{q_{h,t}^2}{q_{h,t-1}}-\sqrt{\frac{q_{h,t}^4}{q_{h,t-1}^2}-\frac{12(1-\rho_{t-1})}{(4-\rho_{t-1})^2}\frac{q_{h,t}^3}{q_{h,t-1}}}\bigg),
\end{align}
before at least one of the firms loses revenue. Algorithm \ref{alg:defection-free cl} ensures that $q_{l,t}$ does not exceed $\hat{q}_{l,t}$. 

It remains to prove that there exist learning rate sequences $\{\alpha_{l,t}\}_t$ and $\{\alpha_{h,t}\}_t$ that respect the constraint $q_{l,t} \leq \hat{q}_{l,t}$. Since improvement in $q_{h}$ increases the revenues of both firms (Prop. \ref{prop:relationship between U and q}), the high-quality firm can set any learning rate schedule $\{\alpha_{h,t}\}_t$ without violating the no-revenue-loss constraints \eqref{eq:high utility increase} and \ref{eq:low utility increase}. For the low-quality firm's learning rate schedule, note that $f_{l}(x)$, as the average of convex functions $f(x;l)$ and $f(x;h)$, is also convex. Therefore,
\begin{align}
    f_{l,t}
    &\geq
    f_{l,t-1} + \nabla_{x_{l,t-1}} f_{l,t-1}^T(x_{l,t}-x_{l,t-1})\\
    &=
    f_{l,t-1} - \alpha_{l,t}\|\nabla_{x_{l,t-1}} f_{l,t-1}\|^2.
\end{align}
Rearranging terms,
\begin{align}
    \alpha_{l,t}
    &\geq
    \frac{f_{l,t-1}-f_{l,t}}{\|\nabla_{x_{l,t-1}} f_{l,t-1}\|^2}\\
    &=
    \frac{q_{l,t}-q_{l,t-1}}{\|\nabla_{x_{l,t-1}} f_{l,t-1}\|^2}.
\end{align}
Therefore, setting $\alpha_{l,t} = \min\bigg\{\frac{\hat{q}_{l,t}-q_{l,t-1}}{\|\nabla_{x_{l,t-1}} f_{l,t-1}\|^2}, 1\bigg\}$ ensures that the low-quality firm's updated quality $q_{l,t}$ does not exceed $\hat{q}_{l,t}$.
\end{proof}

\begin{proof}[Proof of Proposition \ref{prop:sufficient conditions for convergence}]\label{proof:sufficient conditions for convergence prop}
We handle the proof in cases.
\paragraph{Case 1:} $q_{l,0} \leq q_l^*$ and $\rho_0 \geq \rho^*$.

When $\frac{q_{l,t-1}}{q_{h,t}} \geq \rho^*$, the low-quality firm does not update (line 7 of Alg. \ref{alg:defection-free cl}). Once the high-quality firm improves sufficiently so that $\frac{q_{l,t}}{q_{h,t}}=\rho^*$ (note that such a $t$ exists if $q_{l,0} \leq q_l^*$), then convergence is guaranteed. To see this, we use the following lemma.

\begin{lemma}\label{lemma:increasing ratios}
$B(a,b) \geq a$ for all $b \geq 1$. (See Figure \ref{fig:ratio} for pictorial proof.)
\end{lemma}

\begin{figure}[t!]
    \centering
    \includegraphics[height=2.5in]{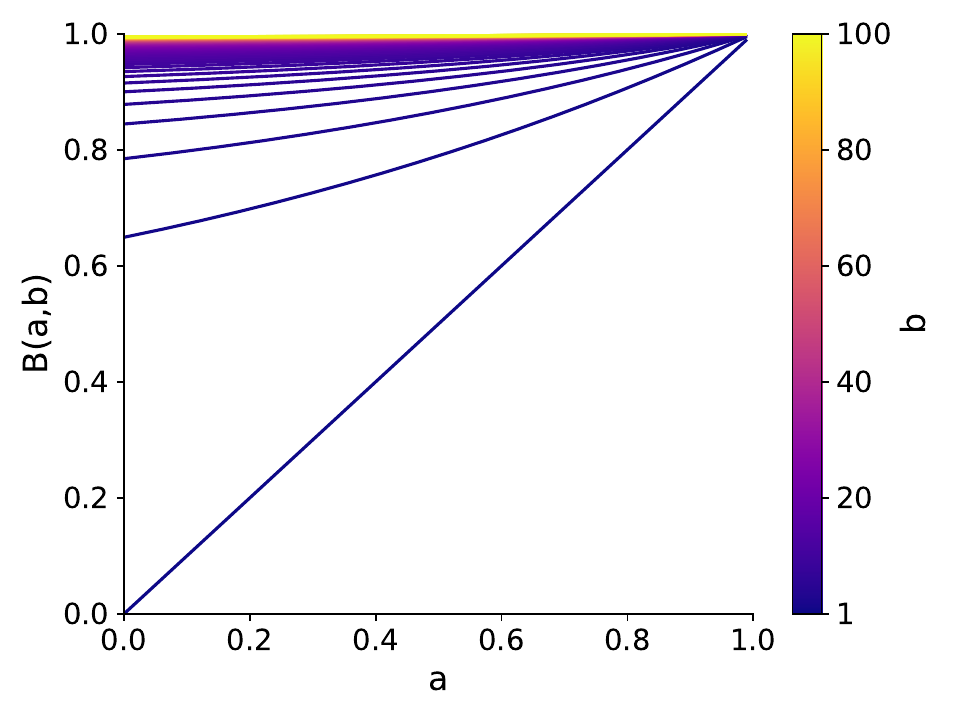}
    \caption{$B(a,b) \geq a$ for all $b \geq 1$.}\label{fig:ratio}
\end{figure}

Consider step $t+1$ at which $\rho_{t}=\frac{q_{l,t}}{q_{h,t}}=\rho^*$. Given the high-quality firm's improvement $q_{h,t} \rightarrow q_{h,t+1}$, if the low-quality firm improves to $q_{l,t+1}=\hat{q}_{l,t+1}$, by Lemma \ref{lemma:increasing ratios}, $\rho_{t+1} \geq \rho_{t}$. Therefore the low-quality firm can always improve to some level $q_{l,t+1} \in [q_{l,t}, \hat{q}_{l,t+1}]$ and ensure that $\rho_{t+1} = \rho^*$ with neither firm losing revenue. Maintaining this improvement schedule, once the high-quality firm improves to $q_h^*$ (using any sequence of learning rates $\{\alpha_{h,t}\}_t$), the low-quality firm will be able to reach $q_l^*$ by observing the constraint in lines 9-11 of Alg. \ref{alg:defection-free cl}.

\paragraph{Case 2:} $q_{l,0} \leq q_l^*$ and $\rho_0 < \rho^*$.

Our strategy for this case will be to show there exist sequences of learning rates $\{\alpha_{h,t}\}_t$ and $\{\alpha_{l,t}\}_t$ such that $\sum_{t=1}^T(\rho_t-\rho_{t-1}) = \rho_T-\rho_0 \geq \rho^* - \rho_0$. We will do this by lower-bounding the quality-ratio gaps $\rho_t-\rho_{t-1} = B(\rho_{t-1},\nicefrac{q_{h,t}}{q_{h,t-1}})-\rho_{t-1}$.

For each $\rho \leq 1$, there is a point (possibly infinite)
\begin{align}
    b_{\rho} 
    &\overset{\text{def}}{=} 
    \max\{b \geq 1: (4-5\rho)\log_{10} b \leq B(\rho, b) - \rho\}\}.
\end{align}
That is, for a given $\rho$, $b_{\rho}$ is the point at which $(4-5\rho)\log b$ goes from being a lower to an upper bound on $B(\rho, b)-\rho$. Define $\tilde{b}$ as the smallest such point over all $\rho \leq 1$, so
\begin{align}
    \tilde{b}
    &\overset{\text{def}}{=}
    \min_{\rho \leq 1} b_{\rho}.
\end{align}
Figure \ref{fig:ratio bounds} plots $b_{\rho}$ for various values of $\rho$ and shows that $\tilde{b} \approx 1.03 = b_{\rho \approx 0.33}$.

\begin{figure}[t!]
    \centering
    \includegraphics[height=2.5in]{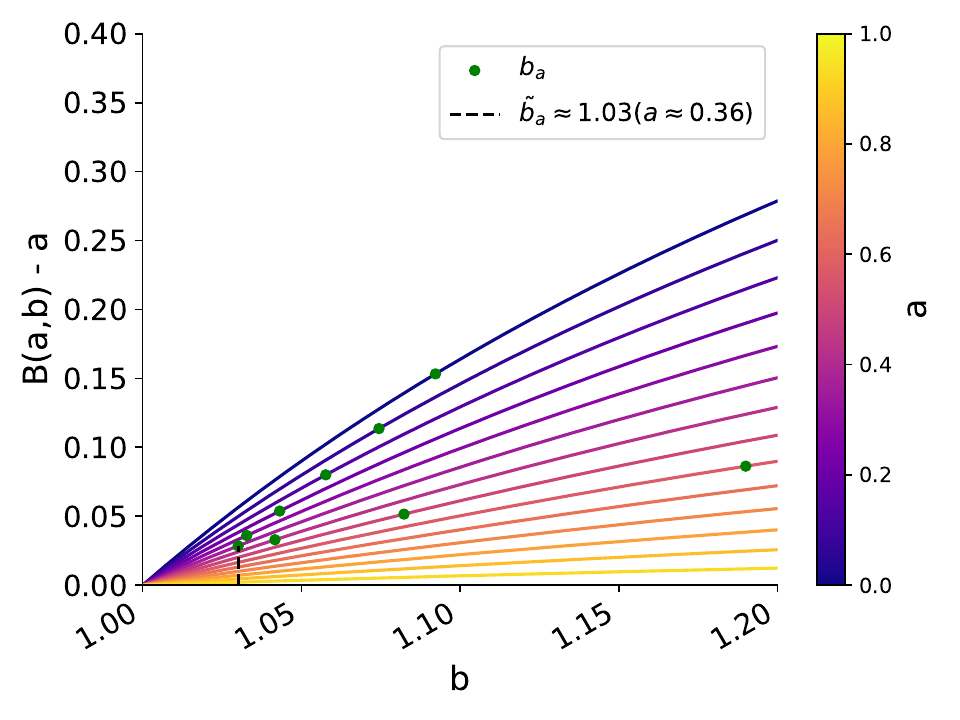}
    \caption{The green dots indicate, for a given $\nicefrac{q_{l,t-1}}{q_{h,t-1}}$ (symbolized by $a$), the upper bound on $\nicefrac{q_{h,t}}{q_{h,t-1}}$ that ensures convergence to the Nash bargaining solution.}\label{fig:ratio bounds}
\end{figure}

By definition of $\tilde{b}$, $(4-5\rho)\log_{10}b \leq B(\rho,b)-\rho$ for any $\rho \leq 1$ and $b \leq \tilde{b}$. Suppose the high-quality firm maintains a learning rate schedule $\{\alpha_{h,t}\}_t$ such that $\nicefrac{q_{h,t}}{q_{h,t-1}} \leq \tilde{b}$ for all $t$ and $T$ is such that $q_h^*-q_{h,T} \leq \epsilon$. Then 
\begin{align}
    \sum_{t=1}^T(\rho_t-\rho_{t-1})
    &=
    \sum_{t=1}^T (B(\rho_{t-1},\nicefrac{q_{h,t}}{q_{h,t-1}}) - \rho_{t-1})\\
    &\overset{(i)}{\geq}
    \sum_{t=1}^T(4-5\rho_{t-1})\log_{10}(\nicefrac{q_{h,t}}{q_{h,t-1}})\\
    &\overset{(ii)}{\geq}
    (4-5\rho^*)\log_{10}(\nicefrac{q_{h,T}}{q_{h,0}})\\
    &\geq
    (4-5\rho^*)\log_{10}(\nicefrac{(q_h^*-\epsilon)}{q_{h,0}}) \label{eq:lower bound on ratio gaps},
\end{align}
where $(i)$ is due to $\nicefrac{q_{h,t}}{q_{h,t-1}} \leq \tilde{b}$, and $(ii)$ is due to 
the fact that $\rho_0 \leq \rho^*$ and Lemma \ref{lemma:increasing ratios}. 

Figure \ref{fig:bounds for convergence} shows that $(4-5\rho^*)\log_{10}(\nicefrac{q_h^*}{q_{h,0}}) \geq \rho^* - \rho_0$, so
\begin{align}
    (4-5\rho^*)\log_{10}\bigg(\frac{q_h^*-\epsilon}{q_{h,0}}\bigg)
    &=
    (4-5\rho^*)\bigg(\log_{10}\bigg(\frac{q_h^*}{q_{h,0}}\bigg)-\log_{10}\bigg(\frac{q_h^*}{q_h^*-\epsilon}\bigg)\bigg)\\
    &\geq
    (\rho^*-\rho_0) - (4-5\rho^*)\log_{10}\bigg(\frac{q_h^*}{q_h^*-\epsilon}\bigg).
\end{align}
Therefore $\rho^* - \rho_T \leq (4-5\rho^*)\log_{10}\bigg(\frac{q_h^*}{q_h^*-\epsilon}\bigg)$.

\begin{figure}[t!]
    \centering
    \includegraphics[height=2.5in]{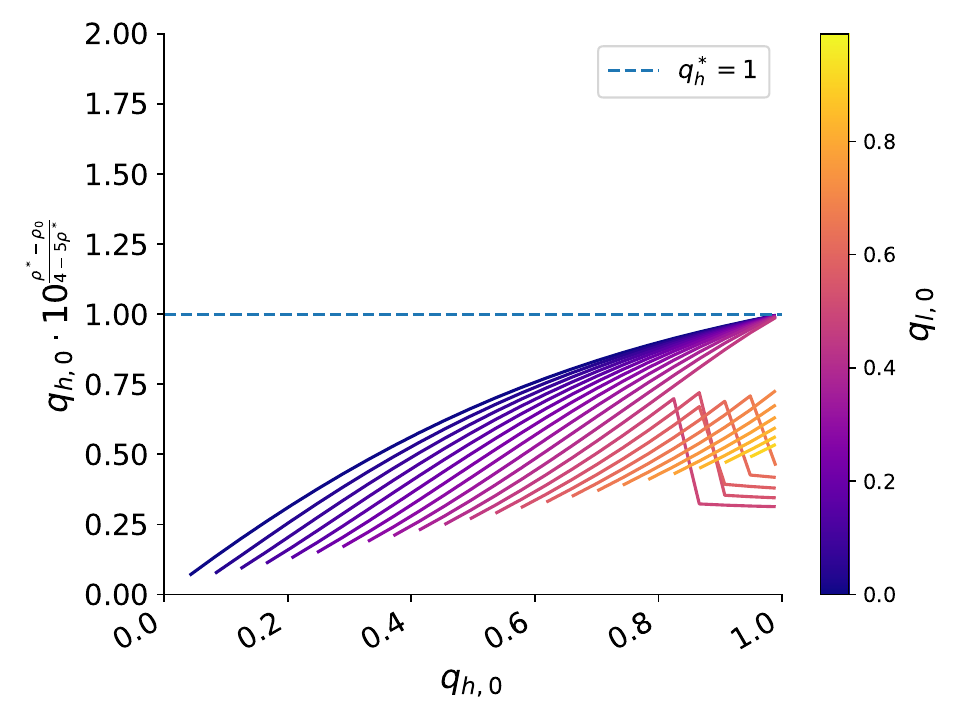}
    \caption{Empirical verification of the inequality: $(4-5\rho^*)\log_{10}(\nicefrac{q_h^*}{q_{h,0}}) \geq \rho^* - \rho_0$}\label{fig:bounds for convergence}
\end{figure}

It remains to show that there exists a sequence of learning rates $\{\alpha_{h,t}\}_t$ such that $\nicefrac{q_{h,t}}{q_{h,t-1}} \leq \tilde{b}$, and $T$ such that $q_h^*-q_{h,T} \leq \epsilon$. Let $\alpha_{h,t} = \min\bigg\{\frac{(\tilde{b}-1)q_{h,t-1}}{\|\nabla_{x_{h,t-1}} f_{h,t-1}\|^2}, \frac{1}{L}\bigg\}$. We analyze what happens when $\alpha_{h,t}$ is each of the values in the \emph{min} expression.

First, suppose $\alpha_{h,t} = \frac{(\tilde{b}-1)q_{h,t-1}}{\|\nabla_{x_{h,t-1}} f_{h,t-1}\|^2}$ for all $t$. $f_h$, as the average of $L$-smooth and convex functions, is also $L$-smooth and convex, so that
\begin{align}
    \frac{q_{h,t-1}+\frac{\alpha_{h,t}}{2}\|\nabla_{x_{h,t-1}} f_{h,t-1}\|^2}{q_{h,t-1}}
    \leq
    \frac{q_{h,t}}{q_{h,t-1}}
    \leq
    \frac{q_{h,t-1}+\alpha_{h,t}\|\nabla_{x_{h,t-1}} f_{h,t-1}\|^2}{q_{h,t-1}}.
\end{align}
Therefore, the choice of $\alpha_{h,t}$ guarantees that $\frac{\tilde{b}+1}{2} \leq \frac{q_{h,t}}{q_{h,t-1}} \leq \tilde{b}$, giving $\frac{q_{h,T}}{q_{h,0}} \geq \bigg(\frac{\tilde{b}+1}{2}\bigg)^T$. From this we see that setting $T \geq \frac{\log(\nicefrac{q_h^*}{q_{h,0}})}{\log(\nicefrac{(\tilde{b}+1)}{2})}$ guarantees convergence to $q_h^*$ in $T'$ steps. 

Now suppose $\alpha_{h,t}=\frac{1}{L}$ for all $t$. Under this condition, standard convergence analysis for gradient descent on convex and $L$-smooth functions gives
\begin{align}
    f_{h,T}-f_h^*
    &\leq
    \frac{L\|x_{h,0}-x_h^*\|^2}{2T}.
\end{align}
Therefore, $f_{h,T}-f_h^* \leq \epsilon$ after $T=\frac{L\|x_{h,0}-x_h^*\|^2}{2\epsilon}$ rounds. 

From the above analysis, we see that after at most $T=\frac{\log(\nicefrac{q_h^*}{q_{h,0}})}{\log(\nicefrac{(\tilde{b}+1)}{2})} + \frac{L\|x_{h,0}-x_h^*\|^2}{2\epsilon}$ rounds, $f_{h,T}-f_h^* = q_h^*-q_{h,T} \leq \epsilon$, completing the proof. 
\end{proof}

\begin{proof}[Proof of Theorem \ref{thm:convergence rate}]\label{proof:convergence rate theorem}
By Taylor's theorem,
\begin{align}
    N(q_l^*,q_h^*)
    &\leq
    N(q_{l,T},q_{h,T}) +
    \frac{\partial N(q_l,q_h)}{\partial q_l}(q_l^*-q_{l,T}) +
    \frac{\partial N(q_l,q_h)}{\partial q_h}(q_h^*-q_{h,T})\\
    &\phantom{{}=1} +
    \bigg(\max_{q_l,q_h}\frac{\partial^2 N(q_l,q_h)}{\partial q_l^2}\bigg)
    \frac{(q_l^*-q_{l,T})^2}{2} + 
    \bigg(\max_{q_l,q_h}\frac{\partial^2 N(q_l,q_h)}{\partial q_h^2}\bigg)\frac{(q_h^*-q_{h,T})^2}{2} \\
    &\phantom{{}=1} + 
    \bigg(\max_{q_l,q_h}\frac{\partial^2 N(q_l,q_h)}{\partial q_h \partial q_l}\bigg)(q_l^*-q_{l,T})(q_h^*-q_{h,T})\\
    &\overset{(i)}{\leq}
    c_1(q_h^*-q_{h,T}) +
    c_2(\rho^*(q_h^*-q_{h,T}) + q_{h,T}|\rho^*-\rho_T|)\\
    &\lesssim
    (q_h^* - q_{h,T}) + |\rho^*-\rho_T|,
\end{align}
where $(i)$ follows from the fact that the gradients of $N$ are bounded by small constants (can be verified with graphing software), qualities $q \in [0,1]$, and $q_l^*-q_{l,T} = \rho^*q_h^* - \rho_Tq_{h,T} \leq \rho^*(q_h^*-q_{h,T}) + q_{h,T}|\rho^* - \rho_T|$.

We now bound $q_h^*-q_{h,T}$. Note that $f_h$, as the average of $L$-smooth and convex functions, is also $L$-smooth and convex. Therefore,
\begin{align}
    f_{h,t}
    &\overset{(i)}{\leq}
    f_{h,t-1} + \bigg(-\alpha_{h,t} + \frac{L\alpha_{h,t}^2}{2}\bigg)\|\nabla_{x_{h,t-1}} f_{h,t-1}\|^2\\
    &\overset{(ii)}{\leq}
    f_{h,t-1} - \frac{\alpha_{h,t}}{2}\|\nabla_{x_{h,t-1}} f_{h,t-1}\|^2\\
    &\overset{(iii)}{\leq}
    f_h^* + \nabla_{x_{h,t-1}} f_{h,t-1}^T(x_{h,t-1}-x_h^*) - \frac{\alpha_{h,t}}{2}\|\nabla_{x_{h,t-1}} f_{h,t-1}\|^2\\
    &=
    f_h^* + \frac{2}{\alpha_{h,t}}(\|x_{h,t-1}-x_h^*\|^2 - \|x_{h,t}-x_h^*\|^2),
\end{align}
where $(i)$ is due to $L$-smoothness of $f_{h}$, $(ii)$ is due to $\alpha_{h,t} \leq \frac{1}{L}$, and $(iii)$ is due to convexity of $f_{h}$. Rearranging terms and summing over $t$,
\begin{align}
    \sum_{t=1}^T \frac{\alpha_{h,t}}{2}(f_{h,t}-f_h^*)
    &\leq
    \sum_{t=1}^T \|x_{h,t-1}-x_h^*\|^2 - \|x_{h,t}-x_h^*\|^2\\
    &\leq
    \|x_{h,0}-x_h^*\|^2\label{eq:sum of iterate gaps}.
\end{align}
Since $\{f_{h,t}\}_t$ are decreasing, \eqref{eq:sum of iterate gaps} implies that
\begin{align}
    f_{h,T} - f_h^*
    &\leq
    \frac{2\|x_{h,0}-x_h^*\|^2}{\sum_{t=1}^T \alpha_{h,t}}.
\end{align}
Noting that $f_{h,T}-f_h^* = q_h^* - q_{h,T}$ completes the proof.
\end{proof}

\begin{proof}[Proof of Corollary \ref{cor:simplified convergence bound}]\label{proof:simplified convergence bound corollary}
Due to Theorem \ref{thm:convergence rate}, showing that $|\rho^*-\rho_T| \leq (4-5\rho^*)\log\bigg(\frac{q_h^*}{q_h^*-\epsilon}\bigg)$ if $T \gtrsim \frac{L\|x_{h,0}-x_h^*\|^2}{\epsilon}$ completes the proof. We handle it in the same cases as in the proof of Proposition \ref{prop:sufficient conditions for convergence}. 
\paragraph{Case 1:} $\rho_0 \geq \rho^*$. From lines 9-11 of Algorithm \ref{alg:defection-free cl}, the low-quality firm will not update its model until after round $T$, where $\rho_{T}=\rho^*$. With only the high-quality firm updating before this point, the firms' qualities will have reached a ratio $\rho^*$ by $T$ steps if $\frac{q_{l,0}}{q_{h,T}} = \rho^*$. Dividing both sides of this equation by $q_{h,0}$ and rearranging terms, $\frac{q_{h,T}}{q_{h,0}} = \frac{\rho_0}{\rho^*}$. As we showed for this case in the proof of Proposition \ref{prop:sufficient conditions for convergence}, $\frac{q_{h,t}}{q_{h,t-1}} \leq \tilde{b}$. Therefore,
\begin{align}
    \frac{q_{h,T}}{q_{h,0}}
    &=
    \frac{\rho_0}{\rho^*}
    \leq
    \tilde{b}^{T},
\end{align}
which gives $T \geq \frac{\log({\nicefrac{\rho_0}{\rho^*}})}{\log(\tilde{b})}$. That is, after $\frac{\log({\nicefrac{\rho_0}{\rho^*}})}{\log(\tilde{b})}$ steps, $\rho_{T} = \rho^*$. As discussed in the proof of Proposition \ref{prop:sufficient conditions for convergence}, the firms can maintain a quality ration of $\rho^*$ for all future rounds, making $|\rho^*-\rho_T|=0$.
\paragraph{Case 2:} $\rho_0 < \rho^*$. As the proof of this case in Proposition \ref{prop:sufficient conditions for convergence} directly shows, $\rho^* - \rho_{T} \leq (4-5\rho^*)\log\bigg(\frac{q_h^*}{q_h^*-\epsilon}\bigg)$ if $T \geq \frac{\log(\nicefrac{q_h^*}{q_{h,0}})}{\log(\nicefrac{(\tilde{b}+1)}{2})} + \frac{L\|x_{h,0}-x_h^*\|^2}{2\epsilon}$. 
\newline
\newline
Combining Cases 1 and 2, if $T \geq \max\bigg\{\frac{\log({\nicefrac{\rho_0}{\rho^*}})}{\log(\tilde{b})}, \frac{\log(\nicefrac{q_h^*}{q_{h,0}})}{\log(\nicefrac{(\tilde{b}+1)}{2})} + \frac{L\|x_{h,0}-x_h^*\|^2}{2\epsilon}\bigg\}$, then $|\rho^*-\rho_T| \leq (4-5\rho^*)\log\bigg(\frac{q_h^*}{q_h^*-\epsilon}\bigg)$, which completes the proof.
\end{proof}

The following lemma gives. 
\begin{lemma}\label{lemma:upper bound on rho^*} For all $\rho_0$ s.t. $\rho_0 \leq \rho^*$, 
$\rho^* \leq 0.43$.
\end{lemma}
\begin{proof}[Proof of Lemma \ref{lemma:upper bound on rho^*}]
The Nash bargaining objective evaluated at $q_h^*=1$ is
\begin{align}\label{eq:nash obj q_h^*}
    N(q_l, q_h^*)
    &=
    \bigg(\frac{q_l(1-q_l)}{(4-q_l)^2} - U_{l,0}\bigg)
    \bigg(\frac{4(1-q_l)}{(4-q_l)^2} - U_{h,0}\bigg),
\end{align}
where $U_{h,0}\overset{\text{def}}{=}U_h(q_{l,0},q_{h,0})$ and $U_{l,0}\overset{\text{def}}{=}U_l(q_{l,0},q_{h,0})$.
Differentiating \eqref{eq:nash obj q_h^*} with respect to $q_l$,
\begin{align}\label{eq:partial q_l nash obj q_h^*}
    &\frac{\partial N(q_l, q_h^*)}{\partial q_l}\\
    &=
    \frac{(7U_{h,0} + U_{h,0}\rho_0 + 4)q_l^3 + (-60U_{h,0}-6U_{h,0}\rho_0+32)q_l^2 + (144U_{h,0}-52)q_l + (-64U_{h,0}+32U_{h,0}\rho_0+16)}{(4-q_l)^5}.
\end{align}
The roots of \eqref{eq:partial q_l nash obj q_h^*} correspond to the roots of the cubic numerator. It can be verified with graphing software that over all starting points $(q_{l,0}, q_{h,0})$ such that $\rho_0 \leq \rho^*$, the roots $q_l^*$ of this cubic  are at most $0.43$. (See Figure \ref{fig:nash max} for empirical evidence.)
\begin{figure}[t!]
    \centering
    \includegraphics[height=2.5in]{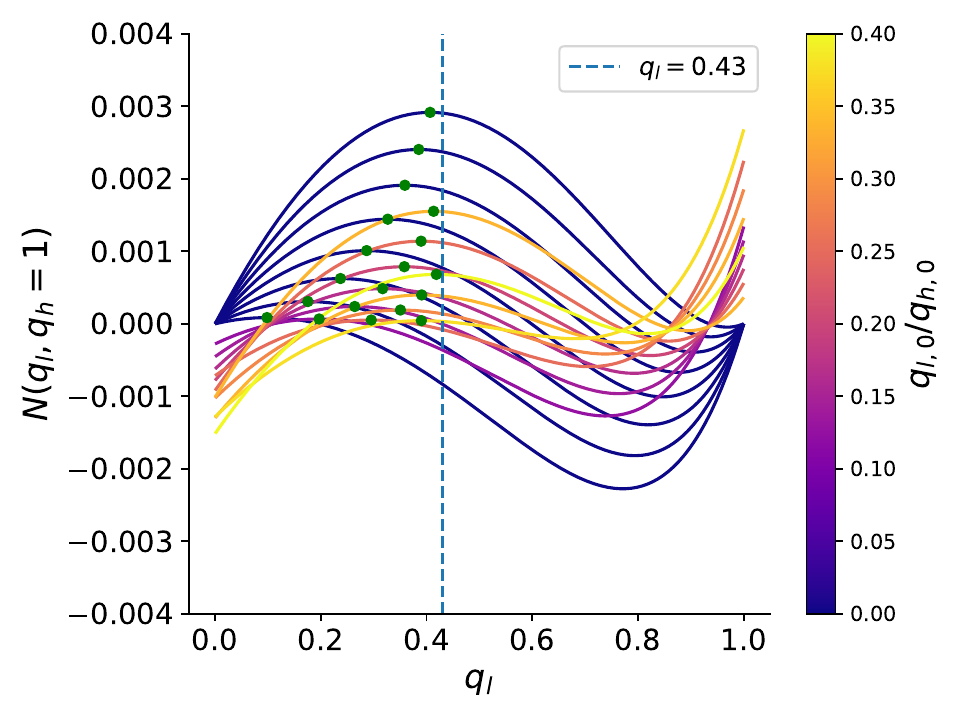}
    \caption{For a range of initial qualities and $q_h=q_h^*=1$, the green dots mark the Nash bargaining solution. The $x$-values of these points are smaller than $0.43$.}\label{fig:nash max}
\end{figure}
\end{proof}
\end{document}